\documentclass[reqno,11pt]{amsproc}
\usepackage{amssymb,a4}

\newtheorem{theorem}{Theorem}
\newtheorem{proposition}{Proposition}
\newtheorem{definition}{Definition}

\newtheorem{corollary}{Corollary}

\theoremstyle{remark}
\newtheorem{remark}{Remark}

\newenvironment{acknowledgement}{\par\medskip\noindent\emph{Acknowledgment.}}

\newcommand\e{\mathrm{e}}

\def\le{\leqslant}

\let\Re\undefined 
\DeclareMathOperator{\Re}{Re}

\DeclareMathOperator{\supp}{supp}

\makeatletter

\newif\ifper\pertrue
\def\per{.}

\newcounter{aucount}
\def\HarvardComma{,}

\newif\ifedplural

\def\au#1#2{{#1 #2}\addtocounter{aucount}{1}}
\def\lau#1#2{{#1 #2},\setcounter{aucount}{0}}

\def\ed#1#2{\ifnum\theaucount=0(\fi{#1 #2}\addtocounter{aucount}{1}}
\def\led#1#2{\ifnum\theaucount=0(\edpluralfalse\else\edpluraltrue\fi{#1
    #2} (\editorname.)):\setcounter{aucount}{0}}
\def\editorname{\ifedplural Eds\else Ed\fi}

\def\et{\ifnum\theaucount=1\else\HarvardComma\fi{} and\ }

\def\bti{\@ifnextchar[\bbti\bbbti}
\def\bbti[#1]#2{\emph{#2}, #1,}
\def\bbbti#1{\emph{#1},}

\def\z{\@ifnextchar[\zz\zzz}
\def\zz[#1]#2#3#4#5{\perfalse{#2} \textbf{#3} (#5), #4 [#1]\per}
\def\zzz#1#2#3#4{{#1} \textbf{#2} (#4), #3\ifper\per\fi\pertrue}

\def\pub{\@ifstar\pubstar\pubnostar}
\def\pubnostar{\@ifnextchar[\@@pubnostar\@pubnostar}
\def\@@pubnostar[#1]#2#3#4{#1, #2, #3, #4\per}
\def\@pubnostar#1#2#3{#1, #2, #3\ifper\per\fi\pertrue}
\def\pubstar[#1]#2#3#4{\perfalse #2, #3, #4 [#1]\per}

\makeatother

\newcommand{\beq}{\begin{equation}}
\newcommand{\eeq}{\end{equation}}
\newcommand{\ba}{\begin{array}}
\newcommand{\ea}{\end{array}}
\newcommand{\bea}{\begin{eqnarray}}
\newcommand{\eea}{\end{eqnarray}}

\newcommand{\Pb}{\mathbb{P}}
\newcommand{\E}{\mathbb{E}}
\newcommand{\R}{\mathbb{R}}

\newcommand{\Z}{\mathbb{Z}}

\begin{document}

\title{Conductivity and the Current-Current Correlation Measure}

\author[J.\ M.\ Combes]{Jean-Michel Combes}
\address{D\'epartement de Math\'ematiques,
Universit\'e du Sud: Toulon et le Var,
F-83130 La Garde, France and CPT, CNRS, Luminy Case 907, Cedex 9, F-13288 Marseille, France}
\email{combes@cpt.univ-mrs.fr}

\author[F.\ Germinet]{Francois Germinet}
\address{D\'epartement de Math\'ematiques,
Universit\'e de Cergy-Pontoise, CNRS UMR 8088, IUF,
F-95000, Cergy-Pontoise, France, and Institut Universitaire de France. F-75005 Paris, France}
\email{germinet@math.u-cergy.fr}

\author[P.\ D.\ Hislop]{Peter D.\ Hislop}
\address{Department of Mathematics,
    University of Kentucky,
    Lexington, Kentucky  40506-0027, USA}
\email{hislop@ms.uky.edu}


\begin{center}
{\it Dedicated to the memory of our friend Pierre Duclos}
\end{center}

\vspace{.1in}

\begin{abstract}
We review various formulations of conductivity
for one-particle Hamiltonians and relate them to the current-current correlation measure.
We prove that the current-current correlation measure for random Schr\"odinger
operators has a density at coincident energies
provided the energy lies in a localization regime.
The density vanishes at such energies and an upper bound
on the rate of vanishing is computed.
We also relate the current-current correlation measure to the localization length.
\end{abstract}

\maketitle
\thispagestyle{empty}


\section{Current-Current Correlation Measure}
\label{cccm-intro}

Higher-order correlation functions for random
Schr\"odinger operators are essential for an understanding of the transport
properties of the system. In this paper, we study the current-current correlation measure
$M( dE_1, dE_2)$ that describes the correlations between electron currents and
is an essential ingredient in the theory of DC conductivity for disordered systems.
The connection between the current-current correlation measure (referred to here as the ccc-measure)
and the DC conductivity
is expressed through the Kubo formula.

We consider a one-particle random, ergodic Hamiltonian $H_\omega$
on $\ell^2 (\Z^d)$ or on $L^2 (\R^d)$. In general $H_\omega$ has the form
$H_\omega =  \frac 12 (- i \nabla - A_0)^2 + V_{per} + V_\omega$, where $A_0$ is a background vector potential,
$V_{per}$ is a real-valued, periodic function, and $V_\omega$ is a real-valued, random potential.
We always assume that $H_\omega$ is self-adjoint on a domain independent of $\omega$.
We restrict $H_\omega$ to
a box $\Lambda \subset \R^d ~\mbox{or} ~\Z^d$ with Dirichlet boundary conditions.
The finite-volume {\it current-current
correlation measures} $M_{\alpha, \beta}^{(\Lambda)}$, for $\alpha, \beta = 1, \ldots, d$,
are the point measures on $\R^2$ defined as follows.
Let $E_m(\Lambda)$ be the eigenvalues of $H_\Lambda$ with normalized eigenfunctions $\phi_m^{(\Lambda)}$.
Let $\nabla_\alpha H_\Lambda \equiv i [ H_\Lambda, x_\alpha] = p_\alpha - (A_0)_\alpha$ be the
$\alpha^{\rm th}$-component of the local velocity operator,
where $p_\alpha = -i \partial / \partial x_\alpha$ is the momentum.
The local velocity operator is independent of the randomness. We denote by
$v_{\alpha; i,j}^{(\Lambda)}$ the matrix elements of the $\alpha^{\rm th}$-component of
the velocity operator in the
eigenstates $\phi_i^{(\Lambda)}$ and $\phi_j^{(\Lambda)}$. That is,
we define
\beq\label{velocity1}
v_{\alpha;i,j}^{(\Lambda)} \equiv
\langle \phi_i^{(\Lambda)}, \nabla_\alpha H_\Lambda \phi_j^{(\Lambda)} \rangle .
\eeq
We now define the {\it finite-volume, current-current correlation measures} for subsets
$\Delta_j \subset \R$ by
\beq\label{ccc1}
M_{\alpha, \beta}^{(\Lambda)} ( \Delta_1, \Delta_2) \equiv \frac{1}{|\Lambda|} \sum_{i,j: E_i (\Lambda) \in \Delta_1;
E_j (\Lambda) \in \Delta_2} v_{\alpha; i,j}^{(\Lambda)}
v_{\beta; j,i}^{(\Lambda)} .
\eeq
These set functions extend to point measures on $\R^2$.
The measures with $\alpha = \beta$, written
$M_\alpha^{(\Lambda)} \equiv M_{\alpha, \alpha}^{(\Lambda)}$, are positive Radon measures on $\R^2$.
We note that $M_{\alpha, \beta}^{(\Lambda)} ( \Delta_1, \Delta_2) =
M_{\beta, \alpha}^{(\Lambda)} ( \Delta_2, \Delta_1)$.
We will denote by $M^{(\Lambda)}$
the positive measure $\sum_{\alpha = 1}^d ~M_{\alpha, \alpha}^{(\Lambda)}$. The existence of the infinite-volume
limits as $\Lambda  \rightarrow \R^d$ of these measures was proved in \cite{[HL],[Pastur0],[Pastur1]}.
We denote the self-averaged, infinite-volume, positive,
Borel ccc-measures by $M_\alpha (dE_1, dE_2)$, $M_{\alpha, \beta} ( dE_1, dE_2)$, and $M(dE_1, dE_2)$.
We remark that these measures can also be defined directly in the
infinite volume using the trace-per-unit volume (see, for example, \cite{[HL]}).

In order to relate the ccc-measures to the conductivity, we need to discuss the Fermi distribution.
The Fermi distribution $n_F(E;T)$ at energy $E$ and temperature $T \geq 0$ is defined for $T>0$ by
\beq\label{fermi1}
n_F (E;T) \equiv (1+ e^{(E-E_F)/ T})^{-1},
\eeq
whereas for $T=0$, it is given by
\beq\label{fermi2}
n_F (E;0) = \chi_{(- \infty, E_F]} (E) .
\eeq
The density matrix $n_F(E_F;H_\omega)$ corresponding to the
$T=0$ Fermi distribution is simply the spectral
projector for $H_\omega$ and the half line $(-\infty, E_F]$.
We will sometimes write $P_{E_F}$ for this projector.

The Kubo formula for the AC conductivity at temperature $T > 0$ and frequency $\nu > 0$
is derived in linear response theory \cite{[BGKS]} (see also \cite{AG,KLM})
by considering a time-dependent Hamiltonian with an external
electric field at frequency $\nu > 0$. This Hamiltonian
can be written as $H_\omega + \mathcal{E} \cdot x \cos \nu t$, although it is more convenient
to study this operator in another gauge. The conductivity
relates the induced current to the electric field $E(x,t) = \mathcal{E} \cdot x \cos \nu t$. A formal
calculation of the current to linear order in the electric field strength $| \mathcal{E}|$ is sketched in section \ref{adiabatic}.
The resulting Kubo formula for the
AC conductivity is:
\bea\label{kubo1}
\lefteqn{ \sigma_{\alpha, \beta}^{AC} ( \nu, T) }  \nonumber \\
 &\equiv &
 \lim_{\epsilon \rightarrow 0}
  \lim_{|\Lambda| \rightarrow \infty}
  \int_{\R^2} \left[ \frac{ n_F(E_1;T) - n_F(E_2;T)}{E_1 -E_2 } \right]
 \delta_\epsilon (E_1 - E_2 + \nu )
 ~M_{\alpha, \beta}^{(\Lambda)} ( dE_1, dE_2)  \nonumber \\
 &= &  - \int_{\R^2}
 ~\left[ \frac{ n_F(E_1;T) - n_F(E_2;T)}{E_1 -E_2 } \right]
~\delta (E_1 - E_2 + \nu  ) ~M_{\alpha, \beta} ( dE_1, dE_2) \nonumber \\
 &=& - \int_{\R}  ~\left[ \frac{ n_F(E + \nu ;T) - n_F(E;T)}{\nu } \right]
 ~M_{\alpha, \beta} ( dE, dE) \ .
\eea
where
\beq\label{delta11}
\delta_\epsilon (s) \equiv  \frac{\epsilon}{ \pi} \frac{1}{ \epsilon^2 + s^2 }.
\eeq
We mention that the ccc-measure has been computed in two models with the aim
of computing the low frequency $ac$ conductivity
\cite{[FP1984],[KP1993]}; see also \cite[section 3.6]{KLM}.

We are interested here in the DC conductivity $\sigma_{\alpha, \beta}^{(DC)}$ at zero temperature.
Formally, this is obtained from (\ref{kubo1}) by taking the limits $\nu \rightarrow 0$
and $T \rightarrow 0$.
We first take $\nu \rightarrow 0$ in (\ref{kubo1}),
to obtain
\beq\label{kubo2}
\sigma_{\alpha, \beta}( 0, T)  = - \int_{\R} \frac{\partial n_F(E;T)}{\partial E}
M_{\alpha, \beta} (dE, dE ) .
\eeq
Next, we take $T \rightarrow 0$. This results in a delta function
$- \delta (E_F - E)$. In order to evaluate the resulting integral,
it is convenient to assume that the infinite-volume
current-current correlation
measure is absolutely continuous with respect to Lebesgue measure on the diagonal
and that it has a locally bounded density $m_{\alpha, \beta}(E,E)$. With this assumption, we obtain
from (\ref{kubo2}),
\beq\label{kubo3}
\sigma_{\alpha, \beta}^{(DC)} (E_F) = m_{\alpha, \beta}(E_F, E_F) .
\eeq
Consequently, the diagonal behavior of the ccc-measure determines the DC conductivity.

In this note, we investigate the existence of the densities $m_{\alpha, \alpha}(E,E)$ and
upper bounds on the rate of vanishing for $E$ in certain energy intervals.
By (\ref{kubo3}), this implies the vanishing of the DC conductivity for energies in these
intervals. The vanishing of the DC conductivity for energies in the localization
regime has been proved by other methods, see section \ref{related1}. We also relate the localization length
to the conductivity, and to the ccc-measure.

We will consider the diagonal behavior of the ccc-measure
at energies in
the complete localization
regime $\Xi^{CL}$.
The complete localization regime is the energy regime discussed by Germinet and Klein \cite{GK2} and
is characterized by strong dynamical localization.
We present precise definitions in the sections below.
We will also discuss a possibly larger energy domain characterized by a finite localization length defined in (\ref{eq:loclength1}).

In order to formulate our results, recall that by the Lebesgue Differentiation
Theorem, the positive measure $M$ has a density at the diagonal point $(E,E)$ if
for any $\epsilon > 0$, with $I_\epsilon (E) = [ E , E + \epsilon ]$, we have that the limit
\beq\label{eq:leb1}
\lim_{\epsilon \rightarrow 0} \frac{ M( I_\epsilon (E), I_\epsilon (E))}{\epsilon^2}
\eeq
exists and is finite. It is necessarily nonnegative.
We say that the resulting density $m (E,E)$ vanishes on the diagonal at the point $(E,E)$ at
a rate given by a function
$g \geq 0$, with $g(s=0) = 0$, if, in addition,
for all $\epsilon >0$ small, we have
\beq\label{defn-rate}
0 \leq \frac{M (I_\epsilon (E), I_\epsilon (E)) }{\epsilon^2} \leq g(\epsilon).
\eeq

The boundedness of the current-current correlation measure on the
diagonal means that the measure has no atoms on the diagonal. For
comparison, the ccc-measure for the free Hamiltonian is easy to
compute due to the fact that the velocity components $\nabla_\alpha
H$ commute with the Hamiltonian. A simple calculation shows that
\beq\label{eq:ccc-free1}
M_{\alpha, \beta} (dE_1, dE_2) = n_0 (E_1)
\delta (E_1 - E_2) \delta_{\alpha, \beta} dE_1 dE_2 ,
\eeq
where
$n_0 (E)$ is the density of states of the free Laplacian at energy
$E$. Consequently, the limit on the diagonal, as described in
(\ref{eq:leb1}) does not exist and the DC conductivity is infinite
at all energies. Boundedness may be obtained simply for energies in
a strong localization regime where (\ref{sl-bounds1}) holds as
stated in Theorem \ref{main1}. If, in addition, a Wegner
estimate (see \cite{CH}, \cite{CHK-2007}) holds, then the DC conductivity vanishes.

We now state our main result on the vanishing of the ccc-density on the diagonal.

\begin{theorem}\label{main1}
Let $E \in \Xi^{CL}$. Then, the current-current correlation measure
density $m(E,E)$ exists and vanishes on the diagonal. If, in addition, the Wegner estimate \eqref{wegner1} holds on $\Xi^{CL}$,
then for any
bounded interval $I_0 \subset \Xi^{CL}$ with $E \in I_0$,
and for any $0 < s < 1$, there is a finite constant $C_{I_0, s} <\infty$ so that
\beq\label{vanish1}
\frac{1}{\epsilon^2}
M (I_\epsilon (E), I_\epsilon (E)) \leq C_{I_0, s} ~\epsilon | \log \epsilon |^{\frac{1}{s}} ,
\eeq
where $I_\epsilon (E) = [ E , E+ \epsilon]$.
\end{theorem}

We note that if, in addition to the Wegner estimate and
the other hypotheses of Theorem \ref{main1}, we know that Minami's estimate
is satisfied in a neighborhood of a given energy, then
the rate of vanishing can be improved for certain intervals near the
diagonal (see Corollary \ref{cor:mott1}).
We discuss this and some generalizations in section \ref{sec:main1-proof} after the proof of the main theorem.

\begin{remark}
Suppose that $E_0$ is a lower band edge at which the IDS exhibits a
Lifshitz tail behavior. For example, if $E_0$ is a band edge of the
deterministic spectrum $\Sigma$ and $E > E_0$, the existence of
Lifshitz tails \cite{Ghribi,Klopp} means that for $E$ sufficiently
close to $E_0$, the IDS $N(E)$ satisfies
\beq\label{ids1}
N(E)-N(E_0) \leq C_{E_0}e^{- \frac{\alpha}{(E-E_0)^{d/2}}}.
\eeq
We write $E(I_\epsilon(E_0))$ for the spectral projector for $H_\omega$ and
the interval $I_\epsilon(E_0)$. It follows from (\ref{ids1}) that for $\epsilon > 0$ small enough,
we have
\beq\label{wegner2}
\E \{ {\rm Tr} \chi_0 E(I_\epsilon
(E_0)) \chi_0 \} \leq C_{E_0} e^{- \frac{\alpha}{\epsilon^{d/2}}},
\eeq
for some constant $\alpha >0$ and where $\chi_0$ is the
characteristic function for the unit cube about zero. We use this
estimate in (\ref{boundI2}) below in place of the usual Wegner
estimate (\ref{wegner1}). As a consequence, we obtain an exponential
rate of vanishing on the diagonal. There is a finite, constant $D_0
> 0$ so that for all $\epsilon >0$ small enough, \beq\label{vanish2}
\frac{1}{\epsilon^2} M (I_\epsilon (E_0), I_\epsilon (E_0)) \leq D_0
e^{- \alpha / \epsilon^{d/2} } . \eeq Of course, this only occurs at
a countable set of energies (the lower band edges). This result, however, stating that,
roughly,  $m(E + \epsilon, E)$ vanishes at an exponentially fast
rate as $\epsilon \rightarrow 0$, is consistent with the Mott theory
of conductivity.
\end{remark}

\begin{remark}
The optimal rate of vanishing in Theorem \ref{main1} is not known. We note, however,
the relationship between the localization length (see section \ref{sec:loclength1})
and the ccc-measure \cite[section V]{BvESB}:
\beq\label{locallength0}
\ell^2 (\Delta) = 2 \int_{\Delta \times \R} ~ \frac{ M (d E_1, d E_2 )}{( E_1 - E_2)^2} ,
\eeq
from which it follows that the ccc-density must vanish faster that $\mathcal{O}(\epsilon)$ on the diagonal
if the localization length is to be finite. We provide a proof of (\ref{locallength0}) as part of Proposition \ref{locallength1}
in section \ref{sec:loclength1}.
\end{remark}

This leads us to our second result relating the localization length $\ell (\Delta)$, for an interval $\Delta \subset \R$,
to the ccc-measure. The localization length bounds the ccc-measure and the vanishing of the localization length implies
the vanishing of the DC conductivity. For this result, we do not need the hypothesis of complete localization.

\begin{theorem}\label{thm:loc-length1}
For $E \in \R$,
and for any $\epsilon > 0$ small, define the interval
$I_\epsilon (E) \equiv [E , E + \epsilon]$.
Suppose that there is a constant $0 \leq M_E < \infty$ so that the localization
length $\ell ( I_\epsilon (E)) \leq M_E < \infty $, for all $\epsilon > 0$ small. Then,
the density of the ccc-measure $m(E, E)$ exists on the diagonal. If, in addition,
 $\ell (I_\epsilon (E)) \rightarrow 0$ as $\epsilon
\rightarrow 0$, then we have
\beq\label{eq:vanish01}
\lim_{\epsilon \rightarrow 0} \frac{M (I_\epsilon (E), I_\epsilon (E))}{\epsilon^2} = 0.
\eeq
Consequently, $\sigma^{(DC)} (E) = 0$.
\end{theorem}

Finally, we mention the open question of the existence of a density $m(E_1, E_2)$
for $E_1 \neq E_2$ that we refer to as the off-diagonal case. This is important in
view of the formula for the localization length (\ref{locallength0}), and in order to control the density in a neighborhood of the diagonal.
Unfortunately, the approach developed in this note does not seem to allow us to control the off-diagonal
behavior.
We mention that
for lattice models on $\ell^2 ( \Z^d)$, Bellissard and Hislop \cite{BH} proved
the existence of a density $m_\alpha (E_1, E_2)$
at off-diagonal energies $(E_1, E_2)$ provided that 1) the energies lie in
a region away from the diagonal $E_1 = E_2$ determined by the strength of the disorder,
and 2) the density of the single-site probability
measure extends to a function analytic in a strip around the real axis.


\section{Existence and Vanishing of the Density on the Diagonal in the Complete Localization Regime}
\label{sec:main1-proof}

In this section, we give a simple proof of Theorem \ref{main1} on the existence of, and vanishing of, the
ccc-density on the diagonal $m(E,E)$ when $E \in \Xi^{CL}$. We denote by $C$ a generic, nonnegative, finite
constant whose value may change from line to line. We work in the infinite-volume framework
using the trace-per-unit volume and refer the reader to \cite{[BGKS]} and \cite{BvESB} for a complete discussion.
Let $(\Omega, \Pb)$ be a probability space with a $\Z^d$-ergodic action $\tau_a: \Omega \rightarrow
\Omega$, $a \in \Z^d$. Let $a \in \Z^d \rightarrow U_a$ be a unitary representation of $\Z^d$
on $\mathcal{H}$. A {\it covariant operator} $A$ on a separable Hilbert space $\mathcal{H}$ is a $\Pb$-measurable, operator-valued
function $A = \{ A_\omega ~|~ \omega \in \Omega \}$ such that for $a \in \Z^D$,
we have $(U_a A U_{-a})_\omega = A_{\tau_{-a} \omega}$. Let $\chi_0$ be the characteristic function
on the unit cube centered at the origin. The {\it trace-per-unit volume}
of a covariant operator $A$, denoted by $\mathcal{T}(A)$, is given by
\beq\label{eq:trace-per-unit1}
\mathcal{T}(A) \equiv \E \{ {\rm Tr} ( \chi_0 A \chi_0 ) \},
\eeq
when it is finite.

We recall that for $\Delta_j \subset \R$, with $j=1,2$, the positive ccc-measure $M_\alpha$
is given by
\bea\label{density1}
M_\alpha
( \Delta_1, \Delta_2) & =  & \mathcal{T} ( \nabla_\alpha H E(\Delta_1) \nabla_\alpha H E(\Delta_2) ) \nonumber \\
   &=& \lim_{|\Lambda| \rightarrow \infty} \frac{1}{|\Lambda|}
   ~{\rm Tr} \chi_\Lambda \nabla_\alpha H E(\Delta_1) \nabla_\alpha H E(\Delta_2) \chi_\Lambda  .
\eea
We recall from \cite[section 3.2]{[BGKS]} that $\mathcal{K}_2$ is the space of all measurable,
covariant operators so that
\beq\label{eq:k2}
\E \{ \| A \chi_0 \|_2^2 \} < \infty.
\eeq
We refer to \cite{[BGKS]} for further definitions and details.
Because the spectral projector $E(\Delta) \in \mathcal{K}_2$, the cyclicity of $\mathcal{T}$
permits us to write
\beq\label{cyclic}
\mathcal{T} (\nabla_\alpha H E(\Delta_1) \nabla_\alpha H E(\Delta_2) )
= \mathcal{T} ( E(\Delta_2) \nabla_\alpha H E(\Delta_1) \nabla_\alpha H E(\Delta_2) ) .
\eeq

Let $\chi_x$ be a bounded function of compact support in a neighborhood of $x \in \R^d$.
Let $\| \cdot \|_p$ denote the $p^{th}$-trace norm for $p \geq 1$. We give a characterization of the
region $\Xi^{CL}$ of complete localization given in
\cite{GK1} and \cite{GK2}.

\begin{definition}\label{stronglocreg1}
An energy $E \in \Sigma$ belongs to $\Xi^{CL}$ if there exists a
open neighborhood $I_E$ of $E$ such that the following bound holds.
Let $\mathcal{I}_E$ denote all functions
$f \in L_0^\infty (\R)$, with $\supp f \subset I_E$ and $\sup_{t \in \R} |f (t)| \leq 1$.
Then, for any
$0 < s < 1$, there is a finite constant $C_s > 0$ so that
the Hilbert-Schmidt bound holds
\beq\label{sl-bounds1}
\sup_{f \in \mathcal{I}_E} \E \{ \| \chi_x f (H_\omega ) \chi_y  \|_2^2 \} \leq C_s
e^{ - \|x-y\|^s}, ~\mbox{for all} ~x,y \in \R^d.
\eeq
\end{definition}

\begin{remark}\label{fmm1}
This estimate can be improved to $s=1$ by the method of fractional moments
\cite{{AENSS},{AM1}}.
\end{remark}

\begin{proof}[Proof of Theorem~\ref{main1}]
\noindent
1. We begin with formula (\ref{cyclic}) and take $\Delta _1
= \Delta_2 = I_\epsilon (E) = [E, E+ \epsilon]$. We write the velocity operator as
a commutator
\beq\label{commutator1}
\nabla_\alpha H =  i [ H, x_\alpha ] = i [ H-E, x_\alpha] .
\eeq
For brevity, we write $P_\epsilon \equiv E(I_\epsilon (E))$. We then have
\bea\label{density2}
M_\alpha ( I_\epsilon (E), I_\epsilon (E)) &=& \mathcal{T} ( P_\epsilon \nabla_\alpha H P_\epsilon
\nabla_\alpha H P_\epsilon ) \nonumber \\
 & = & - \mathcal{T} ( P_\epsilon [ H-E,x_\alpha ] P_\epsilon
[ H-E, x_\alpha ]  P_\epsilon ) .
\eea
Next, we expand the commutators and obtain four terms involving $P_\epsilon$.
This is facilitated by the introduction of a
a function $f_\epsilon (s) \equiv (s - E) \chi_{I_\epsilon (E)} (s) / \epsilon$, so that
$| f_\epsilon | \le 1$. We obtain a factor of $\epsilon$ from each projector so that
\beq\label{density21}
M_\alpha ( I_\epsilon (E), I_\epsilon (E)) = - \epsilon^2 (I + II + III + IV ) ,
\eeq
where we have:
\beq\label{i1}
I   =   \mathcal{T} ( f_\epsilon (H) x_\alpha f_\epsilon (H) x_\alpha
 P_\epsilon ) ;
 \eeq
 \beq\label{iII}
 II = - \mathcal{T} ( f_\epsilon (H) x_\alpha P_\epsilon x_\alpha f_\epsilon (H) ) ;
 \eeq
 \beq\label{iIII}
 III = - \mathcal{T} ( P_\epsilon x_\alpha f_\epsilon^2 (H) x_\alpha P_\epsilon ) ;
\eeq
\beq\label{iIV}
IV = \mathcal{T} ( P_\epsilon x_\alpha f_\epsilon (H) x_\alpha f_\epsilon (H) ).
\eeq

\noindent
2. In order to estimate these four terms for $E \in \Xi^{CL}$, we need the following bounds.
Let $F_{\alpha, \epsilon}(H_\omega)$ denote the operator $x_\alpha f_\epsilon (H_\omega) x_\alpha$.
The bound in Definition \ref{stronglocreg1} implies the following estimate for any integer $q > 0$ and $s$
as in definition \ref{stronglocreg1}:
\beq\label{sl-bound12}
\sup_{f \in I_\epsilon} \E \{ \| (F_{\alpha, \epsilon}(H_\omega) )^{q/2} f(H_\omega) \chi_0 \|_2^2 \} \leq
{C}(s,q, d) \Gamma (( q + d )/s)^2 ,
\eeq
where $\Gamma (t)$ is the gamma function and ${C}(s,q, d)$ depends on $s$ through
the constants $C_s$ and $\alpha_s$ as in (\ref{sl-bounds1}) and grows linearly in
$q$. To prove this bound, we use a partition of unity
on $\R^d$ given by $\sum_{k \in \Z^d} \chi_k = 1$ and we define operators
$A_{k \ell} \equiv \chi_k f(H_\omega) \chi_\ell$ and $B_{k \ell} \equiv
\chi_k (x_\alpha f(H_\omega) x_\alpha )^q \chi_\ell$.
We then estimate the Hilbert-Schmidt norm in (\ref{sl-bound12}) by
\bea\label{eq:sl-bound13}
\| (F_{\alpha, \epsilon}(H_\omega) )^{q/2} f(H_\omega) \chi_0 \|_2^2 &=&
\| \chi_0 f(H_\omega) (F_{\alpha, \epsilon}(H_\omega) )^{q} f(H_\omega) \chi_0 \|_1 \\
&\leq& \sum_{m,k \in \Z^d} \| A_{0 m} \|_2 \| A_{\ell 0} \|_2 \| B_{m \ell} \| .
\eea
Since $\| B_{m \ell} \| \leq C (q) \|m\|^q \| \ell \|^q$, where $C(q)$
is linear in $q$ for integer $q$,
we obtain from (\ref{sl-bounds1}) and (\ref{eq:sl-bound13}),
and the integral representation of the gamma function
\bea\label{eq:sl-bound14}
\E \{ \| (F_{\alpha, \epsilon}(H_\omega) )^{q/2} f(H_\omega) \chi_0 \|_2^2 \}
&\leq & C(s) \left( \sum_{m \in \Z^d} |m|^q e^{-\alpha_s \|m\|^s } \right)^2 \\
 &\leq & C(s,q, d) \Gamma ( (q+ d)/s)^2,
\eea
proving the estimate (\ref{sl-bound12}). We will also use the operator bounds
\beq\label{sl-bound2}
\| f (H) \| \leq 1, ~~|f | \leq 1 ,
\eeq
and by a simple trace class estimate,
\beq\label{sl-bound3}
\mathcal{T} ( P_\epsilon ) \leq C,
\eeq
for some finite constant $C > 0$.

\noindent
3. We now bound the term $I$ in (\ref{i1}) as follows, recalling that $| f_\epsilon | \leq 1$
and that $P_\epsilon f_\epsilon (H_\omega) = f_\epsilon (H_\omega)$:
\bea\label{boundI}
|I| &\leq & | \E \{ {\rm Tr } \chi_0 f_\epsilon (H) x_\alpha f_\epsilon (H) x_\alpha
 P_\epsilon \chi_0 \} | \nonumber \\
  & \leq & \mathcal{T} (P_\epsilon)^{1/2} ~( \E \{ \| F_{\alpha, \epsilon} (H_\omega)
  f_\epsilon (H_\omega) \chi_0 \|_2^2  \} )^{1/2} \nonumber \\
  & = & \mathcal{T} (P_\epsilon)^{1/2} \{ \mathcal{T} ( f_\epsilon (H_\omega) F_{\alpha, \epsilon} (H_\omega)^2
  f_\epsilon (H_\omega) P_\epsilon  )  \}^{1/2} \nonumber \\
  & \leq & \mathcal{T} (P_\epsilon)^{1/2 + 1/4} ~( \E \{ \| F_{\alpha, \epsilon} (H_\omega)^2
  f_\epsilon (H_\omega) \chi_0 \|_2^2  \} )^{1/4}
\eea
Bounds (\ref{sl-bound12}), together with (\ref{sl-bound2})-(\ref{sl-bound3}),
show that the term in (\ref{boundI}) is uniformly bounded as $\epsilon
\rightarrow 0$. It is clear that terms $II-IV$ are bounded in a similar manner.
As a consequence, it follows from these bounds and
(\ref{density21}) that
\beq\label{density3}
\lim_{\epsilon \rightarrow 0 } \frac{M_\alpha ( I_\epsilon (E), I_\epsilon (E))}{\epsilon^2} \leq C < \infty .
\eeq
This implies that the density $m(E,E)$ exists for $E \in \Xi^{CL}$.

\noindent
4. The rate of vanishing can be calculated with more careful estimates. We use the
Wegner estimate in place of (\ref{sl-bound3}) in order to obtain more powers of $\epsilon$. The Wegner estimate
for infinite-volume operators has the form
\beq\label{wegner1}
\mathcal{T} ( E(J) ) = \E \{ {\rm Tr} \chi_0 E(J) \chi_0 \} \leq C_0 |J| ,
\eeq
for any subset $J \subset \R$. Turning to term $I$ in (\ref{boundI}), we iterate the
argument $n$ times, recalling that $P_\epsilon f_\epsilon (H_\omega) = f_\epsilon (H_\omega)$,
and obtain
\beq\label{boundI2}
|I| \leq  \mathcal{T} (P_\epsilon)^{\frac{1}{2} + \ldots +\frac{1}{2^n}} ~\left(
\E \{ \| ( F_{\alpha, \epsilon} (H_\omega))^{2^{n-1}}
  f_\epsilon (H_\omega) \chi_0 \|_2^2  \} \right)^{\frac{1}{2^n}} .
\eeq
We use the strong localization bound (\ref{sl-bound12}) in (\ref{boundI2}), together with the Wegner estimate (\ref{wegner1}) for $I_\epsilon
(E)$, and Stirling's formula for the gamma function, to obtain
\beq\label{Ibound3}
|I| \leq C (s, d) \epsilon^{1 - \frac{1}{2^n}} ~(2^n)^{\frac{1}{s}} .
\eeq
We choose $n$ so that $2^n \sim | \log \epsilon|$ producing the upper bound
$\epsilon |\log \epsilon |^{\frac{1}{s}}$.
The remaining terms can be estimated in a similar manner. \end{proof}



\begin{remark}
We note that if we consider intervals of the form
\beq\label{eq:interval1}
I_\epsilon = [E,E+\epsilon] \times [E-\epsilon,E] \subset \R^2,
\eeq
we can improve the rate of vanishing. This is a consequence of the analysis in \cite{KLM} in the context of Mott's formula for the AC
conductivity. It
requires the Minami estimate in addition to the Wegner estimate.
We recall that Minami's estimate for local operators $H_\omega^\Lambda$ is a second-order correlation estimate.
We assume that the single-site probability measure has a bounded density $\rho$. Let $E_\Lambda ( \Delta)$
denote the spectral projector for $H_\omega^\Lambda$ and the interval $\Delta$. The Minami estimate is:
\beq\label{minami}
\E \left\{{\rm Tr}  E_\Lambda (\Delta) ( {\rm Tr} E_\Lambda (\Delta) -1 )  \right\}
\le (\|\rho\|_\infty |\Delta| |\Lambda|)^2.
\eeq
This was proved by Minami \cite{[Mi]} for lattice models on $\Z^d$, and recent, simplified proofs
have appeared in \cite{[BHS],[CGK],[GV]}. More recently, the Minami estimate
was proved in \cite{CGK2} for certain Anderson models in the continuum for an interval
of energy near the bottom of the deterministic spectrum.

\begin{corollary}\label{cor:mott1}
For any energy $E \in \Xi^{CL}$ for which both the Wegner and Minami estimates hold in an interval $I_0$ containing $E$, we have
\beq\label{eq:mott2}
\frac{1}{\epsilon^2}
M([E,E+\epsilon],[E-\epsilon,E]) \leq C_{I_0} ~\epsilon^2 | \log \epsilon |^{d+2} .
\eeq

\end{corollary}

\begin{proof}
As mentioned, this result is a consequence of the analysis in \cite{KLM} of the Mott formula.
We go back to the decomposition given in \eqref{density21} of that paper,
and estimate each term using \cite[Theorem~4.1]{KLM}.
Applying \cite[Theorem~4.1]{KLM} directly to bound $M([E,E+\epsilon],[E-\epsilon,E])$, we obtain
(\ref{eq:mott2}).
\end{proof}

\end{remark}


\section{Relation to localization length}\label{sec:loclength1}

The second moment of the position operator is used in the following covariant
definition of the localization length due to Bellissard, van Elst, and Schulz-Baldes
\cite{BvESB}.

\begin{definition}
The localization length for an interval $\Delta \subset \R$ is
\beq\label{eq:loclength1}
\ell (\Delta)^2 \equiv \lim_{T \rightarrow \infty} \frac{1}{T} \int_0^T \mathcal{T}
( E(\Delta) | x(t) - x|^2 E(\Delta) ) ~dt,
\eeq
where $x(t) \equiv e^{-itH} x e^{i tH}$.
\end{definition}

If $\Delta \subset \R$ is contained in the complete localization region it is easily
seen from (\ref{sl-bounds1}) that the localization length is finite but one expects that the finiteness
of the localization length holds in a much larger energy domain. Notice, however, that as
shown in \cite{BvESB},
the finiteness of $\ell(\Delta)$ implies that the spectrum of $H_\omega$ in $\Delta$ is pure point.
This does not {\it a priori}
imply estimates such as (\ref{sl-bounds1}).

The ccc-measure on the diagonal and the localization length are very closely related.
Recall that the velocity operator is $\nabla H = i [ H, x]$.

\begin{proposition}\label{locallength1}
Let $E( d\lambda)$ be the spectral family for $H_\omega$. The ccc-measure is related to the localization length by
\beq\label{eq:ccc-loc2}
\ell (\Delta)^2 = 2 \int_\R \int_\Delta \frac{
\mathcal{T} (E (d\mu) \nabla H E (d \lambda) \nabla H E ( d\mu) )}{(\lambda - \mu)^2} .
\eeq
Furthermore, we have the bound
\beq\label{eq:ccc-loc1}
\frac{\mathcal{T}( E(\Delta) \nabla H E (\Delta) \nabla H
E (\Delta) ) }{ |\Delta|^2} \leq  \ell (\Delta)^2 .
\eeq
\end{proposition}

\begin{proof}
We use the Fundamental Theorem of Calculus to write
\beq\label{eq:fund1}
x(t) - x = - \int_0^t ~e^{-i s H} V e^{i sH} ~ds.
\eeq
Taking the absolute square, we obtain
\beq\label{eq:fund2}
| x(t) - x |^2  =  \int_0^t ~dw \int_0^t ~ds ~e^{-i s H} \nabla H e^{i (s- w) H} \nabla H e^{i w H} .
\eeq
Using the spectral family for $H$ in the form $\int_\R E(d \lambda) = 1$, we obtain
\bea\label{eq:fund3}
\lefteqn{\mathcal{T}( E(\Delta) |x(t)-x|^2 E(\Delta) )} \nonumber \\
& = & \int_\Delta \int_\R
\left\{ \int_0^t ~ds e^{- i s (\mu - \lambda)} \int_0^t ~dw  e^{ i w ( \mu - \lambda)}
\mathcal{T} ( E(d\mu) \nabla H E(d \lambda) \nabla H E( d\mu) ) \right\} . \nonumber \\
 & &
\eea
We first perform the integration over $s$ and $w$.
We next take the time average. Finally, taking the limit $T\rightarrow \infty$,
we obtain (\ref{eq:ccc-loc2}).

In order to obtain (\ref{eq:ccc-loc1}),  we write the projectors in the trace-per-unit volume on the left of (\ref{eq:ccc-loc1}) as
$E(\Delta) = \int_\Delta E(d\mu)$ and note that $(\lambda - \mu)^2 \leq | \Delta|^2$
for $\lambda, \mu \in \Delta$.
\end{proof}

The simple {\it proof of Theorem} \ref{thm:loc-length1} follows from (\ref{eq:ccc-loc2}). For $E \in \R$, and $\epsilon >0$ small, we have
\bea\label{eq:loc-lbd1}
\frac{2}{\epsilon^2} M( I_\epsilon (E), I_\epsilon (E)) & = &
\frac{2}{\epsilon^2} \int \int_{I_\epsilon (E) \times I_\epsilon (E)}
 \mathcal{T} ( E( d\mu) \nabla H E( d \lambda ) \nabla H E( d\mu) ) \nonumber \\
  & \leq &   \ell ( I_\epsilon (E) )^2.
\eea
Taking $E_1=E_2$, we obtain the proof of the second part of the theorem.

We next present an explicit formula relating the localization length $\ell (\Delta)$ to the
second moment of the position operator.

\begin{proposition}\label{prop-second1}
We have the following identity:
\beq\label{eq:loc-xmoment1}
\frac{1}{2} \ell (\Delta)^2 = \mathcal{T} (E(\Delta) |x|^2 E(\Delta) ) -
\int_\Delta \mathcal{T}( E (d\lambda) x E (d\lambda) x E(d\lambda )) \geq 0 .
\eeq
where $E( d\lambda) $ is the spectral family for
$H_\omega$.
Consequently, if the second moment of the position operator
$\mathcal{T} (E(\Delta) |x|^2 E(\Delta) ) $ on the right side
of (\ref{eq:loc-xmoment1}) is finite, then the localization length is finite.
\end{proposition}

\begin{proof}
The proof of the equality in (\ref{eq:loc-xmoment1}) follows from
(\ref{eq:ccc-loc2}) by expanding the commutator in the definition of $\nabla H$.
The nonnegativity is obvious and this proves the second part of the proposition.
\end{proof}

We remark that the formulas in (\ref{eq:loclength1}) and in (\ref{eq:loc-xmoment1}) are manifestly covariant with respect to lattice translations.
It is known \cite{BvESB} that $\ell (\Delta) < \infty$ implies that the spectrum of $H_\omega$ in $\Delta$ is pure point almost surely.
The finiteness of the term $\mathcal{T}( E(\Delta) |x|^2 E(\Delta))$ should be compared with the finiteness condition of the Fermi
projector given in (\ref{assumption1}) for the continuum model and in (\ref{fermiproj2}) for the lattice model.



\section{Conductivity for Random Schr\"odinger Operators}\label{related1}


In this section, we discuss the various notions of conductivity that have occurred in recent
literature and provide a justification for the calculation of the conductivity
presented in section \ref{cccm-intro}.
As a starting point, we begin with the adiabatic approach to the Kubo formula as presented in \cite{[BGKS]}.
We then consider related results of Fr\"ohlich and Spencer (1983) \cite{[FS1]}, Kunz (1987) \cite{[Kunz1]}, Bellissard,
van Elst , and Schulz-Baldes (1994) \cite{BvESB}, Aizenman-Graf (1998) \cite{AG},
and Nakano (2002) \cite{Nakano1}.
We write $v_\alpha$ or $\nabla_\alpha H$ for the $\alpha^{\rm th}$-component of the velocity operator
$v_\alpha = \nabla_\alpha H = i [ H, x_\alpha]$.

We mention another derivation of the Kubo-Streda formula for the transverse conductivity
in two-dimensions for perturbations of the Landau Hamiltonian at zero frequency and temperature
is presented in \cite[section IV]{CNP}.
This derivation does not require the use of an electric field that is adiabatically switched on.
On the other hand, the Fermi level has to be restricted to a spectral gap of the Hamiltonian,
rather than to the region of localization.


\subsection{The Adiabatic Definition of the Conductivity}\label{adiabatic}

\subsubsection{Kubo formula through linear response theory}
In the paper of Bouclet, Germinet, Klein and Schenker \cite{[BGKS]}, the total charge transport
is calculated in linear response theory in which the
electric field is adiabatically switched on. This provides a rigorous
derivation of the DC Kubo formula for the
conductivity tensor. A noncommutative integration approach is presented in \cite{DG,D} which somewhat simplifies manipulation of operators.
The random, $\Z^d$-ergodic Hamiltonian $H_\omega$ on
$L^2 (\R^d)$ has the form
\beq\label{general1}
H_\omega = (i \nabla + A_\omega)^2 + V_\omega,
\eeq
where $(A_\omega, V_\omega)$ are random variables so that this operator
is self-adjoint. We add a time-dependent, homogeneous, electric field $\mathcal{E}(t)$ that is
adiabatically switched on from $t= - \infty$ until $t=0$, when it obtains full strength. Such
a field is represented by $\mathcal{E}(t) = e^{\eta t_-} \mathcal{E}$,
where $\mathcal{E}$ is a constant and $t_- \equiv
\mbox{min} ~(0 , t)$. The usual gauge choice is the time-dependent Stark Hamiltonian given by
\beq\label{general2}
\tilde{H}_\omega (t) = H_\omega + \mathcal{E}(t) \cdot x .
\eeq
This Hamiltonian is not bounded from below. As is well-known, and used
in \cite{[BGKS]}, one can make another choice of gauge to eliminate this technical problem. Let
$F(t)$ be the function  given by
\beq\label{general3}
F(t) = \int_{- \infty}^t ~\mathcal{E}(s) ~ds,
\eeq
and note that the integral converges provided $\eta > 0$.

We now consider another Hamiltonian obtained from
$H_\omega$ in (\ref{general1}) by a time-dependent gauge transformation
using the operator $G(t) = e^{i F(t) \cdot x}$:
\beq\label{general4}
H_\omega (t) \equiv G(t) H_\omega G^* (t) =
(i \nabla + A_\omega + F(t))^2 + V_\omega.
\eeq
This operator is manifestly bounded from below almost surely provided $V_\omega$ is also.
Furthermore, the two Hamiltonians
${H}_\omega (t)$ and $\tilde{H}_\omega (t)$, given in (\ref{general2}),
are physically equivalent in that they generate the same dynamics. If
$\psi_t$ solves the Schr\"odinger equation generated by $H_\omega (t)$,
then the gauge transformed wave function
$G(t)^* \psi_t$ solves the Schr\"odinger equation associated with $\tilde{H}_\omega (t)$.
Consequently, we will work with $H_\omega (t)$.

In order to describe the current at time $t=0$, the system is prepared in an initial equilibrium state
$\xi_\omega$ at time $t = - \infty$. The initial state is usually assumed to be the density matrix
corresponding to the Fermi distribution at temperature $T \geq 0$. If
$T > 0$, the state with Fermi energy $E_F$ is given by $\xi_\omega = n_F (H_\omega; T)$,
where $n_F (E; T)$ is given in (\ref{fermi1}). If $T=0$, the initial equilibrium
state is the Fermi projector $\xi_\omega = n_F (H_\omega; 0) = \chi_{(- \infty, E_F]}(H_\omega)$.
We will often write $P_{E_F}$ for this projection.
More general states are allowed, see \cite[section 5]{[BGKS]}.
The crucial assumption for the $T=0$ case is
\beq\label{assumption1}
\E \{ \| |x| P_{E_F} \chi_0 \|_2^2 \} < \infty ,
\eeq
which holds true whenever the Fermi level lies in an interval of
complete localization $\Xi^{CL}$ \cite{GKsudec}.
If the initial state is given by a Fermi distribution, together with
(\ref{assumption1}) if $T=0$,
the authors prove the existence of a unique time-dependent density matrix $\rho_\omega (t)$
solving the following Cauchy problem:
\bea\label{cauchy1}
i \partial_t \rho_\omega (t) &=& [ H_\omega(t), \rho_\omega (t) ]_\ast \nonumber \\
 \lim_{t \rightarrow - \infty} \rho_\omega (t) &=& \xi_{E_F} ,
\eea
in suitable noncommutative ${L}^1$ and  ${L}^2$ spaces
(see \cite{D,DG} for the construction and use of the more natural ${L}^p$ spaces defined over the
reference von Neumann algebra of the problem,
and the corresponding Sobolev spaces). In particular,
the commutator in (\ref{cauchy1}) requires some care in its definition when $H_\omega$
is an unbounded operator. The subscript star
reminds the reader of these noncommutative integration spaces.

In \cite{[BGKS]}, the authors also prove that in the $T=0$ case,
the density matrix $\rho_\omega (t)$ is an
orthogonal projection for all time.
Given this density matrix $\rho_\omega (t)$, the current at time $t=0$ is defined to be:
\beq\label{current1}
J(\eta, \mathcal{E}; E_F ) \equiv \mathcal{T} (\tilde{v} (0) \rho_\omega (0)) ,
\eeq
where the modified velocity operator is $\tilde{v}(0)  \equiv i [ H_\omega, x] - 2 F(0) = v - 2F(0)$.
At time $t = - \infty$, the system is assumed to be in equilibrium so the initial current is zero.
The current $J(\eta, \mathcal{E}; E_F )$ is the net current at time zero
obtained by adiabatically switching on the electric field.

Let $U(s,t)$ be the unitary propagator for the time-dependent Hamiltonian
(\ref{general4}). We denote by $\xi_{E_F} (\tau)$ the time-evolved operator obtained from
the Fermi distribution $n_F (H_\omega (\tau); T)$, as in (\ref{fermi1}) for $T>0$,
or as in (\ref{fermi2}), for $T=0$.
The components of the current are given explicitly by
\beq\label{current11}
J_\alpha (\eta, \mathcal{E}; E_F ) = - \mathcal{T} \left\{ \int_{-\infty}^0 ~d \tau e^{\eta \tau }
\tilde{v}_\alpha (0)
~\left( U(0,\tau) (i [ \mathcal{E} \cdot x, \xi_{E_F} (\tau) ] ) U( \tau, 0) \right)_* \right\}.
\eeq

Linear response theory now states that the conductivity is the constant of proportionality
between the current $J$ and the electric field $\mathcal{E}$, that is, that Ohm's Law holds:
$J = \sigma \mathcal{E}$. This is a tensorial relationship.
Formally, one should expand $J(\eta, \mathcal{E}; E_F)$, given in (\ref{current11}), about $\mathcal{E} = 0$.
The components $\sigma_{\alpha \beta}$ are obtained via the derivative:
\beq\label{current2}
\sigma_{\alpha \beta} ( \eta, E_F)
= \left( \frac{\partial J_\alpha (\eta, \mathcal{E}; E_F)}{\partial \mathcal{E}_\beta} \right)_{ \mathcal{E} = 0 }  .
\eeq
The question of the limit $\mathcal{E} \rightarrow 0$ is addressed in the following theorem.

\begin{theorem}\label{kubo-BGKS}\cite[Theorem 5.9]{[BGKS]}
For $\eta >0$, the current defined in (\ref{current1}) is differentiable with respect
to $\mathcal{E}_j$ at $\mathcal{E} = 0$, and the conductivity tensor is given by
\beq\label{kubo-bgks1}
\sigma_{\alpha \beta} (\eta; E_F) = - \mathcal{T} \left\{  \int_{- \infty}^0 ~d \tau ~e^{\eta \tau} ~ v_\alpha
~ \left( e^{- i \tau H_\omega} (i [ x_\beta, \xi_{E_F} ] ) e^{i \tau H_\omega} \right)_* \right\}.
\eeq
\end{theorem}

\noindent
Once again, the subscript star in (\ref{kubo-bgks1}) indicates that
operators and product of operators are considered within suitable noncommutative integration spaces.
Formula (\ref{kubo-bgks1}) is valid for a large family of
equilibrium initial states including the Fermi distributions.
The analogue of \cite[Eq.~(41)]{BvESB} and
\cite[Theorem~1]{SBB2} then holds.

\begin{corollary}\label{corsgmjk}\cite[Corollary 5.10]{[BGKS]}
Assume that $\mathcal{E}(t)=\Re\e^{i\nu t}$, $\nu\in\R$,
then the conductivity  $\sigma_{jk}(\eta;{{\zeta}};\nu) $ at frequency $\nu$ is given by
\begin{eqnarray}\label{sigmajkbis}
\sigma_{jk}(\eta;\zeta;\nu;0)
&=&\, - \mathcal{T} \left\{ v_\alpha \,
(i\mathcal{L}_1 +\eta +i\nu)^{-1} \left( \partial_\beta \zeta )\right\}\right\rbrace ,
\end{eqnarray}
where $\mathcal{L}_1$ is the Liouvillian generating \eqref{cauchy1}, and $\partial_\beta \zeta=i [ x_\beta, \xi_{E_F} ]$.
\end{corollary}

For $T=0$, following ideas of \cite{BvESB}, the authors recover the Kubo-St\v{r}eda formula for the conductivity
tensor. The distribution $\xi_{E_F}$ is given by the Fermi projector as in (\ref{fermi2}).
In this case, we can take the limit $\eta \rightarrow 0$ and obtain an expression for the DC conductivity.

\begin{corollary}\label{kubo-streda-bgks1}\cite[Theorem 5.11]{[BGKS]}
Under the hypothesis (\ref{assumption1}) on the Fermi projector $P_{E_F}$, the conductivity tensor
is given by
\beq\label{kubo-streda-1}
\sigma^{DC}_{\alpha, \beta} (E_F) = \sigma_{\alpha \beta} = - i \mathcal{T} \left\{ P_{E_F} [[x_\alpha, P_{E_F}],[x_\beta, P_{E_F}]]_* \right\} .
\eeq
The tensor is antisymmetric, so that $\sigma_{\alpha \alpha} = 0$, for $\alpha = 1, \ldots, d$. If,
in addition to (\ref{assumption1}), the magnetic vector potential in (\ref{general1})
vanishes, $A_\omega = 0$, so that the Hamiltonian $H_\omega$
is time-reversal invariant, all components of the conductivity tensor vanish.
\end{corollary}

\noindent
Hence, all components of the DC conductivity vanish for time-reversal invariant systems
(Hamiltonians (\ref{general1}) with $A = 0$)
provided the Fermi projector satisfies the decay hypothesis (\ref{assumption1}).
We note that for lattice models, this condition is
\beq\label{fermiproj2}
\E \left\{ \sum_{x \in \Z^d} ~|x|^2 ~|\langle 0| P_{E_F} | x\rangle |^2 \right\} < \infty ,
\eeq
where $| y \rangle$ denoted the discrete delta function at $y \in \Z^d$.

For systems with nontrivial magnetic fields, the off-diagonal terms may
be nonzero. Of special interest
is the two-dimensional case with a nonzero, constant magnetic field. In this case,
the integer quantum Hall effect is the fact that the off-diagonal term $\sigma_{12}$
is an integer multiple of $e / h$ when the Fermi energy $E_F$
is in the region of localized states (see \cite{BvESB,AG} for lattice models and \cite{GKS,GKS2} for models in the continuum).

\subsubsection{Relation to the ccc-measure}

The formula derived from rigorous linear response theory (\ref{kubo-bgks1}) can be manipulated
in order to derive the $T=0$ formula (\ref{kubo1}).
We write (\ref{kubo-bgks1}) as
\beq\label{kubo-bgks2}
\sigma_{\alpha, \beta} (\eta; E_F) = - 2i
~\int_0^{ \infty} ~d \tau ~e^{- \eta \tau} ~\mathcal{T} ( v_\alpha
 e^{-i \tau H}  [ x_\beta, P_{E_F} ]  e^{i \tau H} ) .
\eeq
We use the spectral theorem $H = \int \lambda dE_\lambda$, and the identity
\beq\label{duhamel}
[ x_\beta, f(H) ] = i \int_{\R} ~\hat{f}(s) ~\int_0^s ~du ~e^{- i(s-u)H} v_\beta e^{-i u H} .
\eeq
Recalling that the Fermi projector $P_{E_F} = n_F (H; 0)$,
we formally write the trace-per-unit volume as
\beq\label{trace1}
\mathcal{T} ( v_\alpha e^{-i \tau H}  [ x_\beta, P_{E_F} ]  e^{i \tau H} )   =
\int \int \frac{n_F(\lambda;0) - n_F(\nu;0)}{\lambda - \nu} \mathcal{T} ( v_\alpha dE_\nu v_\beta  dE_\lambda )
\eeq
Substituting (\ref{trace1}) into (\ref{kubo-bgks2}),
and performing the time integration with use of definition (\ref{delta11}), we obtain
\bea\label{kubo4}
\sigma_{\alpha ,\beta} (E_F) & = &\lim_{\epsilon \rightarrow 0}
\int \int \frac{n_F(\lambda;0) - n_F(\nu;0)}{\lambda - \nu} ~\delta_\epsilon (\lambda - \nu)
\mathcal{T} ( v_\alpha dE_\nu v_\beta  dE_\lambda ) \nonumber \\
 &=& \lim_{\epsilon \rightarrow 0}
\int \int \frac{n_F(\lambda;0) - n_F(\nu;0)}{\lambda - \nu} ~\delta_\epsilon (\lambda - \nu)
M_{\alpha, \beta}(d \nu ,d \lambda ) ,
\eea
where $M_{\alpha,\beta}$ is the ccc-measure.


\subsection{Results on DC Conductivity}

We review several papers concerning the conductivity in one-particle systems. Typically,
these authors assume a formula for the conductivity and then show that under a condition
such as (\ref{assumption1}) or (\ref{fermiproj2}) the DC
conductivity vanishes.

\subsubsection{Fr\"ohlich-Spencer and Decay of the Green's Function}

In their fundamental paper of 1983, Fr\"ohlich and Spencer \cite{[FS1]} proved
the absence of diffusion and the vanishing of the DC conductivity for the lattice
Anderson model.
The authors assumed a form of the Kubo formula
that expresses the DC conductivity of a gas of noninteracting electrons at $T=0$ and Fermi energy $E_F$
in terms of the Green's function. Let $\rho (E) \geq 0$ be the density of states for $H_\omega$.
For $\alpha, \beta = 1, \ldots, d$, the conductivity tensor is given by
\beq\label{conductivityAG1}
\tilde{\sigma}_{\alpha, \beta}^{DC} (E_F) \rho (E_F) = \lim_{\eta \rightarrow 0} \frac{2 \eta^2}{d \pi} \sum_{x \in \Z^d}
x_\alpha x_\beta \E \{ |G(0,x; E_F + i \eta)|^2 \} .
\eeq
Fr\"ohlich and Spencer showed that if $E_F$ lies in an energy
interval for which the multiscale analysis (MSA) holds, then $\sigma^{DC}_{\alpha , \beta} (E_F)
= 0$, assuming $\rho (E_F) > 0$ (see \cite{{HM},{Weg81}}).

\begin{theorem}\cite{[FS1]}
Suppose the Hamiltonian $H_\omega = L + V_\omega$ on $\ell^2 ( \Z^d)$,
where $L$ is the lattice Laplacian (in particular, the magnetic field is zero).
Suppose that bound (\ref{msaest1}) holds for Lebesgue almost all energy in $(a,b)$. Then if $E_F
\in (a,b)$, the DC conductivity defined in (\ref{conductivityAG1})
at Fermi energy $E_F$ vanishes.
\end{theorem}

The key MSA bound on the Green's function is that,
with probability one,
\beq\label{msaest1}
\sup_{\eta > 0} | G_\omega (x,y; E+i \eta)| \leq C_I e^{- m_I | x - y| }.
\eeq
The exponential decay of the Green's function then implies that $\sigma_{\alpha, \beta}^{(DC)}
= 0$. Formula (\ref{conductivityAG1}) is equivalent to (\ref{kubo4}). The following
derivation is presented, for example, in \cite{AG}:
\bea\label{green2}
\lefteqn{ \frac{\eta^2}{\pi} \sum_{x \in \Z^d} x_\alpha x_\beta \E \{ |G(0,x; E_F + i \eta)|^2 \}
} && \nonumber \\
&=& - \frac{\eta^2}{\pi} \E \{ \langle 0| [ x_\alpha, (H-E_F +i \eta)^{-1} ][ x_\beta, (H-E_F -i\eta)^{-1}] |0 \rangle \}
\nonumber \\
 &=& \frac{\eta^2}{\pi} \mathcal{T} \{ R(E_F -i \eta) v_\alpha R(E_F -i \eta) R(E_F +i \eta) v_\beta R(E_F +i \eta) \}
 \nonumber \\
 &=& \int \int \delta_\eta (\lambda - E_F) \delta_\eta ( \nu - E_F) M_{\alpha, \beta} (d \lambda, d\nu) .
\eea
If, as in section \ref{cccm-intro}, we assume the existence of a ccc-density $m_{\alpha, \beta}$,
we may compute the limit $\eta \rightarrow 0$ in (\ref{green2}). We
then obtain $\tilde{\sigma}_{\alpha, \beta}^{DC} (E_F) = m_{\alpha, \beta} (E_F, E_F)$.
So, by (\ref{kubo3}), this gives the same conductivity as
defined in (\ref{kubo1})--(\ref{kubo2}).


\subsubsection{Kunz}

Kunz \cite{[Kunz1]} considered the two-dimensional Landau Hamiltonian with a random potential on $L^2 (\R^2)$
so that $H_\omega = (1/2)(i \nabla + A_0)^2 + V_\omega$, where $A_0 = B(x_2, 0)$ and the velocity
operators $v_j = i [ H_\omega, x_j]$. He assumed
that the conductivity is given by a Kubo formula that he wrote as
\beq\label{kunz1}
\sigma_{1,2} (E_F) = \lim_{\epsilon \rightarrow 0} \frac{1}{i \epsilon} \int_0^\infty e^{- \epsilon t}
\mathcal{T} ( P_{E_F} [ v_1, e^{it H_\omega} v_2 e^{- it H_\omega} ] P_{E_F} ) ~dt,
\eeq
and for the diagonal terms
\beq\label{kunz2}
\sigma_{jj} (E_F) = \lim_{\epsilon \rightarrow 0} \frac{1}{i \epsilon} \int_0^\infty e^{- \epsilon t}
\mathcal{T} ( P_{E_F} [ v_j, e^{it H_\omega} v_j e^{- it H_\omega} ] P_{E_F} ) ~dt + \frac{1}{\epsilon}
\mathcal{T}(P_{E_F}),
\eeq
for $j = 1,2$. Note that \eqref{kunz1} and \eqref{kunz2} follow directly from the above Kubo formula \eqref{kubo-bgks1}, with $ \xi_{E_F}=P_{E_F} $, after integrating by part (and a change of variable $t\to-t$). Assuming band-edge localization for $H_\omega$, he shows that if the disorder is weak enough
relative to $B$ so that there is a gap between the Landau bands and if the Fermi energy $E_F$
lies in the gap between the $n^{th}$ and $(n+1)^{st}$-bands,
then the transverse conductivity $\sigma_{1,2} (E_F)$ is a universal multiple of $(n+1)$.
He also provides arguments for the  localization length to diverge in each Landau band.

\subsubsection{Bellissard, van Elst, and Schultz-Baldes's  Work on Lattice Models}

Bellissard, van Elst, and Schulz-Baldes \cite{BvESB} derived
(\ref{kubo-streda-1}) in a one-electron model using a {\it relaxation time
approximation}. The one-particle Hamiltonian differs from (\ref{general2})
in that a time-dependent perturbation $W_{coll} (t)$ is added that mimics a dissipation process.
The interaction has the form $W_{coll}(t) = \sum_{k \in \Z} W_k \delta (t - t_k)$.
The ordered collision times
$t_k$ are Poisson distributed so that $\tau_k \equiv t_k - t_{k-1}$ are
independent, identically distributed
random variables with an exponential distribution
and mean collision time $\tau = \E (\tau_k)$.
The amplitudes $W_k$ are the collision operators that
are assumed to commute with $H$ and be random operators. The model is discrete.
They computed the time-averaged current using this evolution and found
\beq\label{bvesb1}
J_{\beta, \mu, \mathcal{E}} (\delta) = \sum_{i=1}^2 \mathcal{E}_i
\mathcal{T} ( [x_i, n_F (H;T) ] \frac{1}{ \delta + \tilde{\tau} - \mathcal{L}_H -
\mathcal{E} \cdot \nabla } [ x, H ]) .
\eeq
where $\mathcal{L}_H (A) = i [ H, A]$ is the (bounded) Liouvillian.
The authors then neglected the $\mathcal{E} \cdot \nabla$ term in the resolvent
appearing in (\ref{bvesb1}) and took the limit $\delta \rightarrow 0$. This exists provided
the collision factor $\tilde{\tau} > 0$. This factor is proportional to $1 / \tau$,
the relaxation time. Upon differentiating with respect to the electric field,
they obtained
\beq\label{bvesb2}
\sigma_{\alpha, \beta} = \mathcal{T} ( [x_\alpha,
n_F (H;T) ] \frac{1}{ \tilde{\tau} - \mathcal{L}_H } [ x_\beta , H ]) .
\eeq

In certain cases, the temperature $T$ can be taken to zero and
the relaxation time can be taken to infinity so $\tilde{\tau} \rightarrow 0$.
For example, they proved that for the Landau Hamiltonian with a
random potential in two-dimensions, the off-diagonal conductivity $\sigma^{DC}_{1,2} (E)$ agrees with Kubo-St\v{r}eda formula \eqref{kubo-streda-1} and
is a constant on energy intervals where (\ref{fermiproj2}) holds.
Moreover, it is an integer multiple of $e^2 / h$.

\subsubsection{Aizenman-Graf's Work on Lattice Models}

Stimulated by the integer quantum Hall effect,
Aizenman and Graf \cite{AG} considered the analog of the randomly perturbed Landau
Hamiltonian on the lattice. In this case, the DC conductivity for the system with a constant
magnetic field is given by the Kubo-St\v{r}eda formula:
\beq\label{kubo-str1}
\sigma^{DC}_{\alpha, \beta} (E) = i {\rm Tr} \{ P_E [[ x_\alpha, P_E], [x_\beta, P_E]] \},
\eeq
provided the operators are trace class. The exponential decay of the kernel of the Fermi
projector (\ref{fermiproj1}) implies (\ref{fermiproj2}).
%
%

\subsubsection{Nakano's Result for Lattice Models}\label{nakano1}

The vanishing of the diagonal terms of the DC conductivity
tensor for random, ergodic Schr\"odinger operators
on the lattice $\Z^d$ was also proven by Nakano \cite{Nakano1}
provided the Fermi projector satisfies (\ref{fermiproj2}).

Nakano defined two (scalar) conductivities associated with the $x_1$-direction
(any coordinate direction can be used) differing in the order in which certain limits are taken.
Let $H_{\omega, \mathcal{E}}= H_\omega + \mathcal{E} x_1$,
corresponding to the Anderson Hamiltonian $H_\omega$ with an electric field in the $x_1$-direction.
The current of the system at zero temperature and at time $t$ is given by
\beq\label{current-naka1}
J_1 (t; \mathcal{E}; E_F) =
\mathcal{T} ( e^{i t H_{\omega, \mathcal{E}}}
v_1 e^{- i t H_{\omega, \mathcal{E}}} P_{E_F} ),
\eeq
where $v_1 = i [ H_\omega, x_1]$. The corresponding conductivity is given by
\beq\label{cond-naka1}
\tilde{\sigma}_a (t; E_F) \equiv \lim_{\mathcal{E}\rightarrow 0}
\frac{1}{\mathcal{E}} J_1 (t; \mathcal{E}; E_F).
\eeq
consistent with linear response theory.
The first definition of the conductivity $\sigma_a (E_F)$ is the time average of the conductivity
$\tilde{\sigma}_a (t; E_F)$ defined in (\ref{cond-naka1}):
\beq\label{defn-nakano1}
\sigma_a (E_F) \equiv \lim_{T \rightarrow \infty} \frac{1}{T} \int_0^T ~dt
~\tilde{\sigma}_1 (t; E_F) .
\eeq
The second definition of Nakano replaces  the time-average by an Abel limit:
\beq\label{defn-nakano2}
\sigma_b (E_F) \equiv \lim_{\delta \rightarrow 0} \lim_{\mathcal{E} \rightarrow 0} \frac{1}{\mathcal{E}}
~\int_0^\infty ~dt
\delta e^{-\delta t} ~J_1 (t; \mathcal{E}; E_F) .
\eeq
This is the same as the result derived in \cite{BvESB} using the relaxation time approximation.

\begin{theorem}
Under the assumption (\ref{fermiproj2}) on the Fermi projector $P_{E_F}$,
with $|x|$ replaced by $|x_1|$, we have
\beq\label{nakano-main1}
\sigma_a (E_F) = \sigma_b (E_F) = 0 .
\eeq
\end{theorem}.


\subsection{Localization and the Decay of the Fermi Projector}

When does the hypothesis (\ref{assumption1}) for continuum models, or
(\ref{fermiproj2}) for lattice models, hold?
Aizenman and Graf \cite{AG} used the method of fractional moments, developed in \cite{AM1} and
\cite{A1} for lattice models on $\ell^2 ( \Z^d)$,
to prove the decay of the Fermi projector (\ref{fermiproj2}) provided the Fermi energy
is in the strong localization regime.

The basic hypothesis is that the fractional moment estimate holds:
\beq\label{fmest1}
\sup_{\eta > 0}
\E \{ | G(x,y; E + i \eta )|^s \} \leq C_s e^{-s \mu |x-y| }, ~x,y \in \Z^d,
\eeq
uniformly for $\eta > 0$, for some $0 < s < 1$ and $\mu > 0$.
If (\ref{fmest1}) holds at energy $E_F$, then they prove that the kernel of the Fermi distribution at $T=0$
satisfies
\beq\label{fermiproj1}
\E \{ \langle x | P_{E_F} | y \rangle \} \leq C_0 e^{- \mu |x-y|} .
\eeq
As pointed out by Aizenman and Graf, this exponential decay estimate on the Fermi projector
requires only that $E_F$ lie in a regime of energies for which the fractional moment bound
(\ref{fmest1}) holds. It does not require that (\ref{fmest1}) hold at all energies below
$E_F$.
The fractional moment bound (\ref{fmest1})
is known to hold in the strong localization regime or at extreme energies,
cf.\ \cite{{A1},{AM1}}.

At last, we mention that fast decay of the kernel of the Fermi projection at a given energy $E$  turns out to be equivalent to $E\in\Xi^{CL}$ \cite{GKsudec}.

\begin{acknowledgement}
  PDH thanks JMC for an invitation to CPT and FG for an invitation to Universit\'e de Cergy-Pontoise
  where some of this work was done. The authors thank le Centre Interfacultaire Bernoulli
  and acknowledge partial support from the Swiss NSF during the time this work was completed.
  The authors thank L.\ Pastur for many enjoyable and stimulating discussions. JMC and FG were
  supported in part by ANR 08 BLAN 0261. PDH was supported in part by NSF grant 0803379. We thank the referees
for several helpful comments.
\end{acknowledgement}


\end{document}